\documentclass[a4paper,11pt]{article}

\usepackage{amsmath,amssymb,amsthm}
\usepackage{authblk} 
\usepackage[american]{babel}
\usepackage{balance} 
\usepackage{enumitem} 
\usepackage{float} 
\usepackage[T1]{fontenc}
\usepackage{fullpage} 
\usepackage{graphicx}
\usepackage[utf8]{inputenc}
\usepackage{layout}
\usepackage{mathtools} 
\usepackage{mdframed}
\usepackage{multirow} 
\usepackage{tikz}
\usepackage{titling}
\usepackage{thmtools}
\usepackage{url}
\usepackage[normalem]{ulem} 
\usepackage{xcolor}
\usepackage[backref=page, pagebackref=true, linktocpage, hypertexnames=false]{hyperref} 
\usepackage{algorithm}
\usepackage[noend]{algpseudocode}
\usepackage{cleveref} 
\usepackage{comment}

\usepackage{bm} 

\definecolor{redlink}{rgb}{0.6, 0, 0}

\definecolor{greenlink}{rgb}{0, 0.6, 0}
\definecolor{bluelink}{rgb}{0, 0, 0.6}
\hypersetup{
	colorlinks=true,
	linkcolor=redlink, 
	urlcolor=redlink, 
	citecolor=bluelink, 
}

\theoremstyle{definition}

\newtheorem{question}{Question}

\theoremstyle{plain}
\newtheorem{theorem}{Theorem}[section]

\newtheorem{lemma}{Lemma}[section]

\newtheorem{observation}{Observation}[section]

\theoremstyle{remark}

\newtheorem{claim}{Claim}[section]

\crefname{claim}{claim}{claims}

\setenumerate{itemsep=0mm, topsep=0mm, parsep=0mm}

\makeatletter 
\newcounter{algorithmicH}
\let\oldalgorithmic\algorithmic
\renewcommand{\algorithmic}{%
  \stepcounter{algorithmicH}
  \oldalgorithmic}
\renewcommand{\theHALG@line}{ALG@line.\thealgorithmicH.\arabic{ALG@line}}
\makeatother
\algnewcommand\Break{\textbf{break}}
\algnewcommand\Continue{\textbf{continue}}
\algnewcommand\Exit{\textbf{exit}}
\algnewcommand\Or{\textbf{or~}}
\algnewcommand\And{\textbf{and~}}
\algnewcommand{\IfThenElse}[3]{\State \algorithmicif\ #1\ \algorithmicthen\ #2\ \algorithmicelse\ #3}
\algnewcommand{\IfThen}[2]{\State \algorithmicif\ #1\ \algorithmicthen\ #2}
\algnewcommand{\ForInline}[2]{\State \algorithmicfor\ #1\ \algorithmicdo\ #2}

\newcommand{\eps}{\epsilon} 
\newcommand{\ignore}[1]{} 

\DeclareMathOperator{\poly}{\mathrm{poly}}

\DeclareMathOperator{\port}{port}
\DeclareMathOperator{\tbl}{table}
\DeclareMathOperator{\lbl}{label}
\DeclareMathOperator{\anc}{anc}
\DeclareMathOperator{\lca}{lca}

\DeclarePairedDelimiter\ceil{\lceil}{\rceil}
\DeclarePairedDelimiter\floor{\lfloor}{\rfloor}
\DeclarePairedDelimiter\lAngle{\langle}{\rangle}

\title{Tree-Like Shortcuttings of Trees}

\author{
Hung Le\thanks{University of Massachusetts Amherst. Email: \href{mailto:hungle@cs.umass.edu}{hungle@cs.umass.edu}.}
\and
Lazar Milenković\thanks{Tel Aviv University. Email: \href{mailto:milenkovic.lazar@gmail.com}{milenkovic.lazar@gmail.com}.}
\and
Shay Solomon\thanks{Tel Aviv University. Email: \href{mailto:solo.shay@gmail.com}{solo.shay@gmail.com}.}
\and
Cuong Than\thanks{University of Massachusetts Amherst. Email: \href{mailto:cthan@umass.edu}{cthan@umass.edu}.}}
\date{}

\begin{document}
\maketitle
\begin{abstract}
Sparse \emph{shortcuttings} of trees---equivalently, sparse 1-spanners for tree metrics with bounded \emph{hop-diameter}---have been studied extensively (under different names and settings), since the pioneering works of \cite{Yao82, Cha87, AS87, BTS94}, initially motivated by applications to range queries, online tree product, and MST verification, to name a few. These constructions were also lifted from trees to other graph families using known low-distortion embedding results. The works of \cite{Yao82, Cha87, AS87, BTS94} establish a tight tradeoff between \emph{hop-diameter} and sparsity (or average degree) for tree shortcuttings and imply constant-hop shortcuttings for $n$-node trees with sparsity $O(\log^* n)$. Despite their small sparsity, \textbf{all known constant-hop shortcuttings contain dense subgraphs} (of sparsity $\Omega(\log n)$), which is a significant drawback for many applications.

We initiate a systematic study of constant-hop tree shortcuttings that are ``tree-like''. We focus on two well-studied graph parameters that measure how far a graph is from a tree: \emph{arboricity} and \emph{treewidth}. Our contribution is twofold.
\begin{itemize}
\item \textbf{New upper and lower bounds for tree-like  shortcuttings of trees}, including an \emph{optimal} tradeoff between hop-diameter and treewidth for all hop-diameter up to $O(\log\log n)$. We also provide a lower bound for larger values of $k$, which together yield  $\text{hop-diameter}\times \text{treewidth} = \Omega((\log\log n)^2)$ \emph{for all values of hop-diameter},
resolving an open question of \cite{FL22, Le23post}.
\item \textbf{Applications of these bounds}, focusing on  low-dimensional Euclidean and doubling metrics. A seminal work of Arya et al.\ \cite{ADMSS95} presented a $(1+\epsilon)$-spanner with constant hop-diameter and sparsity $O(\log^* n)$, but with large arboricity.
We show that constant hop-diameter is sufficient to achieve
arboricity $O(\log^*{n})$. Furthermore, we present a $(1+\epsilon)$-stretch routing scheme in the \emph{fixed-port} model with 3 hops and a local memory of $O(\log^2 n / \log\log n)$ bits, 
resolving an open question of \cite{KLMS22}.
\end{itemize}
\end{abstract}
\section{Introduction}

Given a tree $T = (V,E)$ and an integer $k \ge 1$, a \emph{tree shortcutting} of $T$ with hop-diameter $k$ is a graph $G = (V, E')$  such that for every two vertices $u,v \in V$, there is a path in $G$ consisting of at most $k$ edges such that $\delta_G(u,v) = \delta_T(u,v)$, where $\delta_G(u,v)$ represents the distance between $u$ and $v$ in $T$. The problem of constructing sparse tree shortcuttings with small hop-diameter has been studied extensively (under different names and settings) since the pioneering works of \cite{Yao82, Cha87, AS87, BTS94}. The optimal tradeoff is hop-diameter $k$ with $\Theta(n\alpha_k(n))$ edges, where $\alpha_k(n)$ is a very slowly-growing inverse Ackermann function, $\alpha_2(n) = \ceil{\log{n}}$, $\alpha_3(n) = \ceil{\log\log{n}}$, $\alpha_4(n) = \log^* n$, etc. (See \Cref{sec:prelim} for a formal definition.)

Tree shortcutting has many applications, for example, finding max-flow~\cite{Tarjan79}, MST verification~\cite{Tarjan79}, maintaining MST under edge weight increases~\cite{Tarjan79}, computing semigroup product along paths and trees~\cite{Yao82,AS87}, to name a few. In these applications, the shortcut graph serves as a compact data structure where $k$ determines the query time, and the number of edges is the space. Tree shortcutting is a key subroutine in constructing  low-hop $(1+\epsilon)$-spanners for Euclidean~\cite{ADMSS95,Sol13,NS07} and doubling metrics~\cite{CG06,KLMS22}. (See \Cref{sec:prelim} for a formal definition.)

The fundamental drawback of sparse tree shortcuttings is that they are not uniformly-sparse --- all of the aforementioned tree shortcutting constructions contain subgraphs with average degree of $\Omega(\log{n})$ for $k=2$ and $\Omega(\sqrt{n})$ for $k = 3$. As $k$ increases, the (global) sparsity of the shortcutting reduces substantially, but the uniform sparsity remains $\Omega(\log{n})$.   Motivated by this, we initiate a systematic study of low-hop tree shortcuttings that are ``tree-like''. In particular, we are interested in two notions that capture uniform sparsity: \emph{treewidth} and \emph{arboricity}. (See \Cref{sec:prelim} for formal definitions.)

\paragraph{Low-treewidth shortcutting.~} Low-treewidth shortcutting was introduced in the context of low-treewidth embedding of planar metrics. Specifically, Filtser and Le~\cite{FL22} showed that if one can construct a tree shortcutting with hop-diameter $k$ and treewidth $t$, then one obtains an embedding of planar metrics into graphs of treewidth $O(k t/\eps)$ and additive distortion  $+\eps D$, where $D$ is the diameter of the input metric. They constructed a tree shortcutting with  hop-diameter $k=O(\log\log{n})$ and treewidth $t=O(\log\log{n})$, and hence obtained an embedding with treewidth $O((\log\log)^2/\eps)$. This bound does not seem to be optimal --- other works~\cite{FKS19,CCL+24SODA} used different embedding techniques to remove the dependency on $n$ but at the cost of a (much) higher dependency on $1/\eps$. If one can construct a shortcutting with hop-diameter $k = O(1)$ and treewidth $t = O(\alpha_k(n))$ (which matches the sparsity bound), then one gets an embedding of planar metrics with treewidth $O(\alpha_k(n)/\eps) = O(\log^*(n)/\eps)$, which is almost as good as $O(1/\eps)$.
The linear dependency on $1/\eps$ is optimal~\cite{FL22}, whereas the other embedding techniques~\cite{FKS19,CCL+24SODA} have an inherent $\Omega(1/\eps^2)$ barrier. Motivated by this observation, Filtser and Le~\cite{FL22} asked:

\begin{question}[\cite{Le23post,FL22}]\label{q:treewidth}
Is there a tree shortcutting with treewidth $t$ and hop-diameter $k$ such that $k\cdot t = o((\log\log{n})^2)$? 
Furthermore, is it possible for $k \cdot t$ to approach $O(1)$?
\end{question}

In this work, we seek to answer a {\bf broader question}: {\em What is the exact tradeoff between the hop-diameter $k$ and the treewidth $t$ of tree shortcuttings?} Understanding the full tradeoff would not only answer \Cref{q:treewidth} but could also be useful for other applications that require a different hop-diameter. Specifically, in this paper, we use our construction for $k = 3$ to resolve an open problem posed by Kahalon et al.~\cite{KLMS22} regarding compact routing schemes.

\paragraph{Low-arboricity shortcutting.~} Given a set of points  $P$ in an Euclidean (or doubling) space of dimension $d$, one can construct a sparse $(1+\eps)$-spanner with a small hop-diameter. This is achieved via tree covers. (See \Cref{sec:prelim} for a definition.) Known tree cover constructions for Euclidean and doubling metrics achieve stretch $1+\eps$ using $\eps^{-O(d)}$ trees \cite{ADMSS95,BFN19,CCLMST24socg}.
By constructing a $k$-hop shortcutting $G_T$ for each tree $T$ in the cover $\mathcal{T}$, the union $\cup_{T\in \mathcal{T}}G_T$ is a $k$-hop $(1+\eps)$-spanner. For a constant $\eps$ and $d$, the sparsity is only $O(1)$ times worse than the sparsity of the shortcuttings, which is $O(\alpha_k(n))$.  However, if $G_T$ has low treewidth for every tree $T$, the union  $\cup_{T\in \mathcal{T}}G_T$  might have a very large treewidth. A recent work on low-treewidth Euclidean spanners by  Buchin, Rehs, and Scheele~\cite{BRS25} showed that treewidth-$t$ spanners must have stretch $\Omega(n/t^{d/(d-1)})$. That is, if the stretch is $O(1)$, the treewidth has to be very large: $\Omega(n^{1-1/d})$. 

While the treewidth grows substantially under taking union, the arboricity does not: if $G_T$ has arboricity at most $\beta$ for every $T\in \mathcal{T}$, then $\cup_{T\in \mathcal{T}}G_T$ has arboricity $|\mathcal{T}|\beta$, which is  $O(\beta)$ for constant $\eps$ and $d$. This makes arboricity an appealing tree-like measure. Specifically, one could hope to construct a low-arboricity Euclidean
(or doubling) $(1+\eps)$-spanner with small hop-diameter using low-arboricity tree shortcuttings.

\subsection{Our contribution}
Our key \textbf{conceptual contribution} is in identifying the power of tree-like shortcuttings of trees, in particular in the regime of very small hop-diameter. We employ the low-treewidth construction with hop-diameter 3 to improve the state-of-the-art on low-hop compact routing schemes in tree and doubling metrics. Using the low-arboricity shortcuttings, we devise a $(1+\eps)$-spanner with arboricity $O(\log^*{n})$ and hop-diameter 6 for doubling metrics. We also fully resolve \Cref{q:treewidth} in the negative, which is our \textbf{main technical contribution}.
\paragraph{Treewidth.~} Our main technical contribution is a lower bound on the treewidth $t$ of tree shortcuttings for any hop-diameter $k$, which holds even when the underlying graph is a path. This result fully resolves \Cref{q:treewidth} in the negative. 

\begin{restatable}{theorem}{twLB}\label{thm:twLB}
For every  $n \ge 1$, every shortcutting with hop-diameter $k$ for an $n$-point path  must have treewidth:
\begin{itemize}
\item $t = \Omega(k\log^{2/k}{n})$ for even $k$ and $t=\Omega(k(\frac{\log{n}}{\log\log{n}})^{2/(k-1)})$ for odd $k \ge 3$, whenever $k \le \frac{2}{\ln(2e)}\ln\log{n}$;
\item $t = \Omega((\log\log{n})^2/k)$ whenever $k > \frac{2}{\ln(2e)}\ln\log{n}$.
\end{itemize} 
\end{restatable}

Specifically, \Cref{thm:twLB} implies that $t\cdot k = \Omega((\log\log n)^2)$, and therefore, the artificially looking upper bound by Filtser and Le~\cite{FL22} is indeed optimal. 

We prove \Cref{thm:twLB} by constructing a large clique minor $K_{t}$ in any given shortcutting $H$ with hop-diameter $k$. Such a clique minor is a certificate that the treewidth of $G$ is at least $t$. Our lower bound instance is an $n$-vertex path, and the proof is by induction on $k$. We observe that when $k = 2$, shortcutting $H$ contains the clique ${K}_{\floor{\log{n}}+1}$ as a minor. For a higher even value of $k\geq 4$, we use the induction hypothesis for $k-2$. We divide the path into equally sized subpaths and consider two cases. If there is a subpath where every vertex has an edge going out of it, we can readily construct a large minor using this subpath. Otherwise, each subpath $P$ has at least one vertex without an edge going out of $P$. This allows us to use the induction hypothesis for $k-2$ on a graph obtained by contracting each subpath into a single vertex.
The argument for odd values of $k$ is analogous. Although the method of constructing the minor is simple to describe, the analysis is much more intricate, specifically handling the dependency between $k$ and $n$ in the recurrences. 
See \Cref{sec:tw} for details. 

We also give an upper bound construction that matches the lower bound when $k = O(\log\log n)$, which arguably is the most interesting regime.  Our construction is a rather natural adaptation of the known result by Filtser and Le~\cite{FL22}. The difference is that when we shortcut a set of ``cut vertices'' recursively, we set the hop bound to be $k-2$. This requires some technical work; the proofs are deferred to \Cref{sec:twUB}. The following theorem summarizes the construction.

\begin{restatable}{theorem}{twUB}\label{thm:twUB}
For every $n \ge 1$ and every $k = O(\log\log{n})$, every $n$-vertex tree admits a shortcutting with hop-diameter $k$ and treewidth $O(k\log^{2/k}{n})$ for even $k$ and $O(k(\frac{\log{n}}{\log\log{n}})^{2/(k-1)})$ for odd $k \ge 3$.
\end{restatable}

In \Cref{sec:routing-ub}, we use this construction to devise a 3-hop routing scheme for tree metrics with stretch 1 and $O(\log^2{n}/\log\log{n})$ bits per vertex. This answers the open question from \cite{KLMS22}: ``Whether or not one can use a spanner of larger (sublogarithmic and preferably constant) hop-diameter for designing compact routing schemes with $o(\log^2 n)$ bits is left here as an intriguing open question.'' 
In \Cref{sec:routing-lb}, we prove that the bound is tight. Using tree covers, we obtain a 3-hop routing scheme with stretch $(1+\eps)$ that uses $\eps^{-\tilde{O}(d)}\log^2{n}/\log\log{n}$ bits per vertex for every $n$-point metric with doubling dimension $d$. (See \Cref{sec:routing-ub}.) 

\begin{restatable}{theorem}{routingUBdoubling}\label{thm:routingUBdoubling}
For every $n$ and every $n$-point metric with doubling dimension $d$, there is a 3-hop routing scheme with stretch $(1+\eps)$ that uses $\eps^{-\tilde{O}(d)}\cdot\log^2{n}/\log\log{n}$ bits per vertex.
\end{restatable}

This provides the first routing scheme in Euclidean and doubling metrics, where the number of hops is $o(\log n)$, and the labels consist of $o(\log^2{n})$ bits. 

\paragraph{Arboricity.} In \Cref{sec:arb} we focus on low-arboricity shortcuttings. 
We show how to construct shortcuttings of \emph{hierarchically separated trees (HSTs)} where the arboricity grows proportionally to an inverse Ackermann function. 
Our starting point is a sparse shortcutting of an $n$-vertex path. We use the well-known fact about arboricity --- if the edges of a given graph can be oriented so that the in-degree is bounded by $\beta$, then the arboricity is at most $\beta+1$. Our key observation is that one can split every edge $(u,v)$ of a sparse shortcutting into two edges $(u,w)$ and $(w,v)$ in order to reduce the in-degrees of $u$ and $v$. By carefully choosing vertex $w$ for every edge $(u,v)$, we show that the arboricity can be bounded by $\alpha_{k/2+1}(n)$ when the hop-diameter is $k$. Roughly speaking, we achieve the same sparsity bound with a twice as large hop-diameter, which is the consequence of splitting each edge into two. Our construction for paths is then readily applicable to HSTs.
Using the known HST cover \cite{FL22lso}, we obtain the following theorem. See \Cref{sec:arb} for details.

\begin{restatable}{theorem}{arbUBdoubling}\label{thm:arbUBdoubling}
Let $k$ be an even integer and let $\eps \in (0,1/6)$ be an arbitrary parameter. Then, for every positive integer $n$, every $n$-point metric with doubling dimension $d$ admits a $(1+\eps)$-spanner with hop-diameter $k$ and arboricity $\eps^{-O(d)}\alpha_{k/2+1}(n)$.
\end{restatable}

This significantly strengthens the construction of Arya et al.~\cite{ADMSS95}, by providing a uniformly sparse (rather than just sparse) construction with constant hop-diameter: For hop-diameter $k$, we transition from sparsity  $\eps^{-O(d)}\alpha_{k}(n)$ to arboricity  $\eps^{-O(d)}\alpha_{k/2+1}(n)$. In particular, we get arboricity $O(\log^* n)$ with a hop-diameter of 6. Recall that, in contrast to arboricity, one cannot achieve similar dependencies on the treewidth. 

For general trees, we obtain a $k$-hop shortcutting with arboricity  $O(\log^{12/(k+4)}{n})$. Note that this is an improvement over the treewidth bound given in \Cref{thm:twUB}. Specifically, when $k  = \Theta(\log \log n)$, the arboricity is $O(1)$, while the treewidth is $\Omega(\log\log n)$. Using this result and the known tree cover constructions \cite{CCL+23FOCS, CCL+24SODA, BFN19, MN07}, we obtain low-hop spanners for planar, minor-free, and general metrics with small arboricity. See \Cref{sec:arb} for details.

\section{Preliminaries}\label{sec:prelim}

\paragraph{Treewidth.} A tree decomposition of $G=(V,E)$ is a tree $T$ where each node is associated with a subset of $V$ called a bag. The bags $X_1, \ldots, X_\ell$ satisfy: (\emph{i}) $\cup_{i=1}^\ell X_i = V$, (\emph{ii}) for every vertex $v \in V$, the collection of bags containing it forms a connected subtree of $T$, and (\emph{iii}) for every edge $(u,v) \in E$, there is a bag containing it.
The width of $T$ is $\max_{i=1}^\ell|X_i| - 1$; the treewidth of $G$ is the minimum width among all possible tree decompositions of $G$. 
If a graph $G$ contains $K_h$ as a minor, then the treewidth of $G$ is at least $h-1$.

\paragraph{Arboricity.}
The arboricity of $G=(V,E)$  is defined as $\max\ceil{\frac{|E(H)|}{|V(H)|-1}}$, 
where the maximum is taken over all subgraphs $H$ of $G$ with at least two vertices. 

\paragraph{Spanners.} Given a metric space $M_X=(X,\delta)$, which we view as a complete graph $(X, \binom{X}{2})$, a \emph{$t$-spanner} of $M_X$ is a subgraph $H=(X,E)$ with $E \subseteq \binom{X}{2}$, such that for every two vertices $x,y \in X$, their distance in $H$ is at most $t \cdot \delta(x,y)$. The path $P_{x,y}$ realizing this distance is called a $t$-spanner path. The \emph{hop-diameter} of a spanner is the minimum $k$ such that between every two vertices $x,y \in X$ there is a $t$-spanner path with at most $k$ edges.

\paragraph{Tree covers.} 
Given a metric space $M_X = (X,\delta)$, a tree cover of $M_X$ with stretch $t$ is a collection of trees such that for every two vertices $x,y \in X$ we have $\delta(x,y) \le \min_{T\in\mathcal{T}} \delta_T(x,y)\le t\cdot \delta(x,y)$, where $\delta_T(x,y)$ is the distance between $x$ and $y$ in $T$.

\paragraph{Ackermann functions.}
For all $k \ge 0$, the Ackermann functions $A(k,n)$ and $B(k,n)$ are defined as follows:
\[
A(0,n) := 2n, \quad \text{for all } n \ge 0,
\]
\[
A(k,n) :=
\begin{cases}
1 & \text{if } k \ge 1 \text{ and } n = 0, \\
A(k-1, A(k, n-1)) & \text{if } k \ge 1 \text{ and } n \ge 1,
\end{cases}
\]
\[
B(0,n) := n^2, \quad \text{for all } n \ge 0,
\]
\[
B(k,n) :=
\begin{cases}
2 & \text{if } k \ge 1 \text{ and } n = 0, \\
B(k-1, B(k, n-1)) & \text{if } k \ge 1 \text{ and } n \ge 1.
\end{cases}
\]

For all $k \ge 0$ and $n \ge 0$, the inverse Ackermann function $\alpha_k(n)$ is defined as follows:
\begin{align*}
&\alpha_{2k}(n) := \min \{s \ge 0 : A(k,s) \ge n\} \\
&\alpha_{2k+1}(n) := \min \{ s \ge 0 : B(k,s) \ge n \}
\end{align*}
One can easily see that
$\alpha_0(n) = \ceil{n/2}$, $\alpha_1(n) = \ceil{\sqrt{n}}$, $\alpha_2(n) = \ceil{\log{n}}$, $\alpha_3(n) = \ceil{\log\log{n}}$, $\alpha_4(n) = \log^\ast n$, $\alpha_5(n) = \floor{\frac{1}{2} \log^\ast n}$, etc.

\paragraph{Compact routing schemes. }
A compact routing scheme is a distributed algorithm that can deliver packets between any source node and any destination node in the network. Each node in the network has its own \emph{routing table}, which stores local routing information, as well as a unique \emph{label}. During a preprocessing phase, the network is initialized so that each node is assigned a routing table and a label. In the \emph{labeled model}, labels are chosen by the designer (typically of size $\mathrm{poly}(\log n)$), while in the \emph{name-independent model}, labels are chosen adversarially.

Each packet carries a \emph{header} containing the label of the destination and possibly additional auxiliary information. When forwarding a packet destined for vertex $v$, a node $u$ consults its routing table together with the label of $v$ to determine the outgoing edge (specified by a port number) along which the packet should be sent. In the \emph{designer port model}, port numbers are assigned during preprocessing, whereas in the \emph{fixed port model}, port numbers are assigned adversarially. This forwarding process continues until the packet reaches its destination.

A routing scheme has \emph{stretch} $t$ if, for every source--destination pair, the path length taken by the scheme is at most $t$ times the length of a shortest path in the network.

In this paper, we consider the underlying network to be a metric space. The algorithm first chooses an \emph{overlay network} on which routing is performed. The main objective is to minimize the size of the routing tables stored at each vertex.
\section{Low treewidth shortcuttings}\label{sec:tw}
In this section, we show tight tradeoff between treewidth and hop-diameter for tree shortcuttings. In particular, the upper bound (\Cref{thm:twUB}) is proved in \Cref{sec:twUB} and the matching lower bound (\Cref{thm:twLB}) is proved in \Cref{sec:twLB}.

The following two claims are used in the proofs of both theorems.

\begin{claim}
There is an absolute constant $\gamma$ such that for $\alpha \in \{0, 1\}$, every integer $k \ge 4$, every $x > 1$ where the expression is defined, it holds
\begin{equation}
\frac{2}{k} < x^{2/k} - \left(x- \left(\frac{k}{k-2} \right)^{(k-2)/2}\cdot x^{(k-2)/k} - \alpha \right)^{2/k} < \frac{\gamma}{k}\label{eq:identity}
\end{equation}
\end{claim}
\begin{proof}
We rewrite the expression as follows.
\begin{align}\label{eq:rewrite}
x^{2/k}\left(1-\left(1- \left(\frac{k}{k-2} \right)^{(k-2)/2}\cdot x^{-2/k} - \alpha x^{-1} \right)^{2/k}  \right)
\end{align}

Using Maclaurin expansion, we have that $(1+y)^{2/k} = 1 + \frac{2}{k}y + \frac{2-k}{k^2}\cdot (1+\zeta)^{\frac{2}{k}-2}\cdot y^2$, where $\zeta$ is a number between $0$ and $y$. We set $y = -\left(\frac{k}{k-2} \right)^{(k-2)/2}\cdot x^{-2/k} - \alpha x^{-1}$.

\begin{gather}\label{eq:maclaurin}
\left(1-\left(\frac{k}{k-2} \right)^{(k-2)/2}\cdot x^{-2/k} - \alpha x^{-1} \right)^{2/k} =\\
1 - \frac{2}{k} \left (\left(\frac{k}{k-2} \right)^{(k-2)/2}x^{-2/k}  + \alpha x^{-1}\right) - \frac{k-2}{k^2} (1 + \zeta)^{\frac{2}{k}-2}\left( \left( \frac{k}{k-2}\right)^{(k-2)/2}\cdot x^{-2/k} + \alpha x^{-1}\right)^2
\end{gather}

Plugging \Cref{eq:maclaurin} into \Cref{eq:rewrite}, we obtain the following.
\begin{gather*}
x^{2/k}\left(1-\left(1- \left(\frac{k}{k-2} \right)^{(k-2)/2}\cdot x^{-2/k} -\alpha x^{-1}\right)^{2/k}  \right) = \\
\frac{2}{k}\left ( \left(\frac{k}{k-2}\right)^{(k-2)/2} + \alpha x^{-(k-2)/k} \right) +
\frac{k-2}{k^2} (1 + \zeta)^{\frac{2}{k}-2} \left( \left(\frac{k}{k-2}\right)^{(k-2)/2}\cdot x^{-1/k} + 
\alpha x^{-(k-1)/k}\right)^2
\end{gather*}
The lower bound in \Cref{eq:identity} holds because $-1 < y < \zeta < 0$ and $x > 1$. Next we prove the upper bound. 
\begin{gather*}
\frac{2}{k}\left ( \left(\frac{k}{k-2}\right)^{(k-2)/2} + \alpha x^{-(k-2)/k} \right) +
\frac{k-2}{k^2} (1 + \zeta)^{\frac{2}{k}-2} \left( \left(\frac{k}{k-2}\right)^{(k-2)/2}\cdot x^{-1/k} + 
\alpha x^{-(k-1)/k}\right)^2 <\\
\frac{2(e +x^{-(k-2)/k}) + (ex^{-1/k} + x^{-(k-1)/k})^2}{k} 
\end{gather*}
The right-hand side is decreasing with $x$ in the whole domain and we can upper bound it by taking $x=1$. 
\begin{align*}
\frac{2(e +x^{-(k-2)/k}) + (ex^{-1/k} + x^{-(k-1)/k})^2}{k}  < \frac{(e+1)(e+3)}{k} 
\end{align*}
Letting $\gamma = (e+1)(e+3)$, the upper bound from \Cref{eq:identity} follows.
\end{proof}

\begin{claim}\label{clm:bound_ell}
For every $3 \le k \le \frac{2}{\ln(2e)}\ln\log{n}$, it holds that $(\frac{k}{k-2})^{(k-2)/2} (\log{n})^{(k-2)/k} \le (\log{n})/2$.
\end{claim}
\begin{proof}
We have that $k \le \frac{2}{\ln(2e)}\ln\log{n}$. Rearranging the last inequality, we have that $e(\log{n})^{(k-2)/k} \le (\log{n})/2$. The proof is completed by observing that $(\frac{k}{k-2})^{(k-2)/2}$ is monotonically increasing for $k \ge 3$ and $\lim_{k \to \infty}(\frac{k}{k-2})^{(k-2)/2} = e$.
\end{proof}

\begin{claim}\label{clm:taylor}
For every $x \ge 1$ and $k \ge 4$, $x^{\sqrt{(k-2)/k}} \le x - \frac{x\ln{x}}{k}+\frac{x\ln^2{x}}{k^2}$
\end{claim}
\begin{proof}
Let $p\coloneqq 1/k$. Then, $x^{\sqrt{(k-2)/k}}=x^{\sqrt{1-2p}}$. Using Taylor series around $p=0$, we have that $x^{\sqrt{1-2p}} = x-px\ln{x} + R(p)$. The term $R(p)$ has the following form for $0<\xi<p$.
\begin{align*}
R(p) = \frac{p^2}{2}\cdot x^{\sqrt{1-2\xi}}\ln{x}\left(\frac{\ln{x}}{1-2\xi}-(1-2\xi)^{-3/2} \right) \le\frac{p^2}{2}\cdot x^{\sqrt{1-2\xi}}\ln{x} \cdot \frac{\ln{x}}{1-2\xi} \le p^2x\ln^2{x}
\end{align*}
To finish the proof, we replace $p$ by $1/k$.
\end{proof}
\subsection{Upper bound}\label{sec:twUB}
In this section we prove \Cref{thm:twUB}.

\twUB*

The following lemma will be used in proving the theorem.

\begin{lemma}[Cf. Lemma 1 in \cite{FL22}]\label{tree-partition}
Given a parameter $\ell \in \mathbb{N}$ and an $n$-vertex tree $T$, there is a set $X$ of at most $\frac{2n}{\ell + 1} - 1$ vertices such that every connected component $C$ of $T \setminus X$ is of size at most $\ell$ and has at most 2 outgoing edges towards $X$. Furthermore, if $C$ has outgoing 
edges towards $x, y \in X$, then necessarily $x$ is an ancestor of $y$, or vice versa.
\end{lemma}

\subsubsection{Hop-diameter 2}
\begin{lemma}\label{lem:treewidth-2}
For every tree $T$ there is a shortcutting $H_2$ with hop-diameter 2 and treewidth $O(\log{n})$.
\end{lemma}
\begin{proof}
Consider the following recursive construction due to \cite{Sol13,AS24,NS07} which produces a set of edges $E_2$ of $H_2$. 
Take a centroid vertex $v$ of $T$ and add an edge between $v$ and every other vertex of $T$ to $E_2$. Recurse on each subtree of $T \setminus v$. Stop whenever $T$ is a singleton. 

Let $\mathcal{T}_1, \ldots, \mathcal{T}_g$ be the tree decompositions of shortcuttings constructed for the subtrees in $T \setminus v$. Create a new bag $B$ containing only the vertex $v$ and add $v$ to every bag in each $\mathcal{T}_1, \ldots, \mathcal{T}_g$. The tree decomposition of $H_2$ is obtained by connecting $B$ to the roots of each of $\mathcal{T}_1, \ldots, \mathcal{T}_g$. The treewidth satisfies recurrence $W_2(1) = 0$ and $W_2(n) = W_2(n/2) + 1$, which has solution $W_2(n) = O(\log{n})$.

Consider any two vertices $u,v \in T$. Let $w$ be the centroid vertex used in the last recursive call where $u$ and $v$ belonged to the same tree. By construction, $H_2$ contains edges $(w,u)$ and $(w,v)$, meaning that there is a 2-hop path between $u$ and $v$. Vertex $w$ is contained on the path between $u$ and $v$ in $T$, meaning that the stretch of this path is 1.
\end{proof}

\subsubsection{Hop-diameter 3}

\begin{lemma}\label{lem:treewidth-3}
For every tree $T$ there is a shortcutting $H_3$ with hop-diameter 3 and treewidth $O(\log{n}/\log\log{n})$.
\end{lemma}
\begin{proof}
Let $\ell_3 = \log{n}/\log\log{n}$. This value will not change across different recursive levels when the subtree sizes shrink. 

Consider the following recursive construction for constructing the edge set of $E_3$ of $H_3$. Let $X$ be a set of vertices for $T$ as in \Cref{tree-partition} with parameter $\ell=n/\ell_3$ so that $|X| = O(\ell_3)$. 
Connect the vertices of $X$ by a clique and add those edges to $E_3$. 
Next, do the following for every subtree $T'$ in $T \setminus X$. Let $u$ and $v$ be two vertices from $X$ to which $T'$ has outgoing edges. Connect $u$ and $v$ to every vertex in $T'$ and add the edges to $E_3$. Proceed recursively with $T'$. The base case occurs whenever the size of the tree is at most $\ell_3$. In the base case, we connect the vertices by a clique. 

Let $\mathcal{T}_1, \ldots, \mathcal{T}_g$ be the tree decompositions of shortcuttings constructed for the trees in $T \setminus X$. For every $\mathcal{T}_i$, let $u$ and $v$ be the two vertices from $X$ adjacent to the corresponding subtree $W_i$ in $T \setminus X$. (\Cref{tree-partition} guarantees that there is at most two such vertices; it is possible that $u = v$.) Add $u$ and $v$ to each bag in $\mathcal{T}_i$. Construct a new bag $B$ containing all the vertices in $X$ and connect it to the roots of each $\mathcal{T}_1, \ldots, \mathcal{T}_g$. The treewidth of $H_3$ satisfies $W_3(n) \le n$ for $n \le \ell_3$ and $W_3(n) \le W_3(n/\ell_3) + 2$ for $n > \ell_3$. Recall that $\ell_3$ is fixed and does not change across different levels of recursion. The recurrence satisfies $W_3(n) = O(\log{n}/\log\log{n})$.

Consider any two vertices $u, v \in T$. Let $X$ be the set of vertices in the last recursive call when $u$ and $v$ were in the same tree. If $u$ and $v$ are considered in the same base case, then there is a direct edge between them. Otherwise, let $W_u$ (resp., $W_v$) be the subtree in $T \setminus X$ and let $u'$ (resp., $v'$) be the vertex in $X$ that is incident on $W_u$ (resp., $W_v$). By construction, $H_3$ contains edges $(u, u')$, $(v, v')$, and $(u', v')$. (The cases where $u'=u$ or $v'=v$ are handled similarly.)
\end{proof}
\subsubsection[Hop-diameter k≥4]{Hop-diameter \bm{$k\ge 4$}}
\begin{lemma}
For every tree $T$ and every $k\ge 4$ there is a shortcutting $H_k$ with hop-diameter $k$ and treewidth $O(k\log^{2/k}{n})$ for even values of $k$ and $O(k(\frac{\log{n}}{\log\log{n}})^{2/(k-1)})$ for odd values of $k$.
\end{lemma}
\begin{proof}
Let $\ell_k$ be such that $\log{\ell_k}= (\frac{k}{k-2})^{(k-2)/2}(\log{n})^{(k-2)/k}$ for even values of $k$ and $\log{\ell_k}= (\frac{k}{k-2})^{(k-2)/2}(\log{n}/\log\log{n})^{(k-2)/k}$ for odd values. By \Cref{clm:bound_ell}, we have that $\ell_k \le \sqrt{n}$.

Consider the following recursive construction for constructing the edge set of $E_k$ of $H_k$. Let $X$ be a set of vertices for $T$ as in \Cref{tree-partition} such that $|X| = O(\ell_k)$ and every component of $T \setminus X$ has size at most $n/\ell_k$.
Do the following for every subtree $T'$ of $T \setminus X$. Let $u$ and $v$ be two vertices from $X$ to which $T'$ has outgoing edges. Connect $u$ and $v$ to every vertex in $T'$ and add the edges to $E_k$. Proceed recursively with $T'$. The base case  occurs whenever the size of the tree is at most $\ell_k$. In the base case, we connect the vertices of the considered tree using the construction with hop-diameter $k-2$. 
To interconnect the vertices in $X$, we construct an auxiliary tree $T_X$; use the recursive construction with parameter $k-2$ on $T_X$ and add all these edges to $H_k$. This concludes the description of $H_k$.

Next, we analyze the treewidth of $H_k$. Let $\mathcal{T}_1, \ldots, \mathcal{T}_g$ be the tree decompositions of shortcuttings constructed for the trees in $T \setminus X$. For every $\mathcal{T}_i$, let $u$ and $v$ be the two vertices from $X$ adjacent to the corresponding subtree $W_i$ in $T \setminus X$. (\Cref{tree-partition} guarantees that there is at most two such vertices; it is possible that $u = v$.) Add $u$ and $v$ to each bag in $\mathcal{T}_i$.  Let $\mathcal{T}_X$ be the tree decomposition of $T_X$, where $T_X$ is defined as in the previous paragraph. Since $\mathcal{T}_X$ is a valid tree decomposition and $T_X$ contains an edge $(u,v)$, then $\mathcal{T}_X$ contains a bag  $L_{u,v}$ with both $u$ and $v$. Connect the root of $\mathcal{T}'$ to $L_{u,v}$. This concludes the description of the tree decomposition $\mathcal{T}$ of $T$. The treewidth of $\mathcal{T}$ satisfies the recurrence  $W_k(n) = W_{k-2}(n)$ for $n \le \ell_k$ and $W_k(n) \le \max(W_{k-2}(\ell_k), W_k(n/\ell_k) + 2)$ otherwise.

We show that the recurrence satisfies $W_k(n) \le k\log^{2/k}n$ for even values of $k$. (The proof for odd values is similar.) For the base case we use $n \le \ell_k$, where $W_k(n) \le W_{k-2}(n)$. For the induction step, we assume that the hypothesis holds for all the values smaller than $n$.  First, note that $W_{k-2}(\ell_k) \le (k-2)(\log{\ell_k})^{2/(k-2)} = k(\log{n})^{2/k}$. 
\begin{align*}
W_k(n) &\le \max(k(\log{n})^{2/k}, W_k(n/\ell_k) + 2)   \\
&\le \max(k(\log{n})^{2/k}, k(\log{n} - \log{\ell_k})^{2/k} + 2)\\
& \le \max\left(k(\log{n})^{2/k}, k\left(\log{n} - \left(\frac{k}{k-2}\right)^{(k-2)/2}\cdot (\log{n})^{(k-2)/k}\right)^{2/k} + 2\right)\\
& \le \max\left(k(\log{n})^{2/k}, k(\log{n})^{2/k}\right)
\end{align*}
The last inequality follows from the left-hand side part of \Cref{eq:identity} by setting $x = \log{n}$ and $\alpha=0$.

To argue that $H_k$ is a shortcutting with hop-diameter $k$, consider the last step of the construction where both $u$ and $v$ were in the same tree $T$. If $u$ and $v$ were considered in the same base case, then $T$ is equipped with a recursive construction with parameter $k-2$. By induction, there is a shortcutting path between them with hop-diameter $k-2$. Otherwise, let $W_u$ (resp., $W_v$) be the subtree in $T \setminus X$ and let $u'$ (resp., $v'$) be the vertex in $X$ that is incident on $W_u$ (resp., $W_v$). By construction, there is a shortcutting path between $u'$ and $v'$ with at most $k-2$ hops. (The cases where $u'=u$ or $v'=v$ are handled similarly.)
\end{proof}
\subsection{Lower bound}\label{sec:twLB}
We show the treewidth lower bound for shortcutting of the path $L_n = \{1,2,\ldots, n\}$ with hop-diameter $k$. Due to the inductive nature of our argument, we prove a stronger version of the statement, which considers shortcuttings that could potentially use points outside of the given path.

\twLB*

We do so by arguing that any shortcutting of $L_n$ has a large minor. It is well known that a graph with $K_t$ as a minor has treewidth at least $t-1$.

We prove lower bounds for even $k$; the lower bound for odd $k$ can be shown using the same argument.
\subsubsection{Hop-diameter 2} \label{sec:lb-k2-1}
\begin{lemma}
For every $n\ge 1$ and every $n$-vertex path $L_n$, every shortcutting with hop-diameter 2 has ${K}_{\floor{\log{n}}+1}$ as a minor.
\end{lemma}
\begin{proof}
We prove the claim by complete induction over $n$. For the base case, we use $n = 1$, where the claim holds vacuously.

Let $H$ be a shortcutting for $L_n$ with hop-diameter 2. Split $L_n$ into two consecutive parts, $L_1$ and $L_2$, of sizes $\floor{n/2}$ and $\ceil{n/2}$, respectively. Let $H_1$ and $H_2$ be subgraphs of $H$ induced on $L_1$ and $L_2$, respectively. From the induction hypothesis, $H_1$ and $H_2$ have $K_{\log{\floor{n/2}}}$ and $K_{\log{\ceil{n/2}}}$ as minors, respectively. 

Consider the case where every point of $L_1$ has an edge in $H$ to some point in $L_2$. Then ${K}_{\log{\floor{n/2}}}\cup \{H_2\}$ induces a clique minor of size $\log\floor{n/2} + 1$.
Consider the complementary case where $L_1$ has a point $p$ that does not have an edge in $H$ to any point in $L_2$. Then, every point in $L_2$ has a neighbor in $L_1$ because $H$ has hop-diameter $2$. 
Thus, ${K}_{\log{\ceil{n/2}}}\cup \{H_1\}$ induces a clique minor of size $\log\ceil{n/2} + 1$. Hence, the minor size satisfies recurrence $W_2(n) = W_2(\floor{n/2})+1$, with a base case $W_2(1) = 1$. The solution is given by $W_2(n) = \floor{\log{n}} + 1$.
\end{proof}

\subsubsection[Hop-diameter k≥4]{Hop-diameter \bm{$k\ge 4$}}
We give a proof for even values of $k$ such that $4 \le k \le \frac{2}{\ln(2e)}\ln \log{n}$ in \Cref{lem:tw-lb-small-k}. The proof for odd values is analogous. The proof for $k > \frac{2}{\ln(2e)}\ln \log{n}$ is given in \Cref{lem:tw-lb-large-k}.
\begin{lemma}\label{lem:tw-lb-small-k}
For every $n \ge 1$, every even $4 \le k \le \frac{2}{\ln(2e)}\ln \log{n}$, and every $n$-vertex path $L$, every shortcutting with hop-diameter $k$ has  treewidth at least $c_1k\log^{2/k}n-1$, where $c_1$ is an absolute constant.
\end{lemma}
\begin{proof} Let $\ell_k$ be such that $\log{\ell_k}= (\frac{k}{k-2})^{(k-2)/2}(\log{n})^{(k-2)/k}$ for even values of $k$.  (The proof for odd values is similar. There, we choose $\ell_k$ so that $\log{\ell_k}= (\frac{k}{k-2})^{(k-2)/2}(\log{n}/\log\log{n})^{(k-2)/k}$.) By \Cref{clm:bound_ell}, we have that $\ell_k \le \sqrt{n}$, whenever $k \le \frac{2}{\ln(2e)}\ln \log{n}$. 

Split $L$ into consecutive sets of points $L_1,L_2,\ldots, L_{\ell_k}$ 
of size $\floor{n/\ell_k}$ each and ignore the remaining points. Let $H$ be a shortcutting with hop-diameter $k$ for $L_n$. Our goal is to show that the size of a clique minor of $H$ can be lower bounded by the following recurrence.
\begin{equation}
\begin{split}
W_k(n)\ge \min(W_{k-2}(\ell_k),W_k({n}/(2\ell_k))+1)\label{eq:lb-rec-k} \text{ and } 
W_k(1) = 1
\end{split}
\end{equation}

We prove the statement by complete induction over $n$ and $k$. For the base case, we take $n=1$, and $W_k(1) = 1 > c_1k\log^{2/k}n-1$. 

We say that a point in $L_i$ is \emph{global} if it has an edge to a point outside $L_i$ and \emph{non-global} otherwise. We say that $L_i$ is global if all of its points are global and \emph{non-global} otherwise. We consider two complementary cases as follows.

\paragraph{Case 1: Every $L_i$ is non-global.} For every $L_i$ and $L_j$ the path between a non-global point in $L_i$ and a non-global point in $L_j$ must have the first (resp., last) edge inside $L_i$ (resp., $L_j$).
Let $H'$ be obtained from $H$ by contracting each $L_i$ into a single vertex. (Clearly, $H[L_i]$ is connected.) Let $L'$ be the path obtained from $L_n$ by contracting every $L_i$ into a single point.
Then $H'$ is a $(k-2)$-hop shortcutting of $L'$ with stretch 1. Thus, $H'$ has a minor of size $W_{k-2}(\ell_k) \ge c_1(k-2)\log^{2/(k-2)}\ell_k -1=c_1k\log^{2/k}{n}-1$.
\begin{figure}[h]
    \centering
    \includegraphics[width=\linewidth]{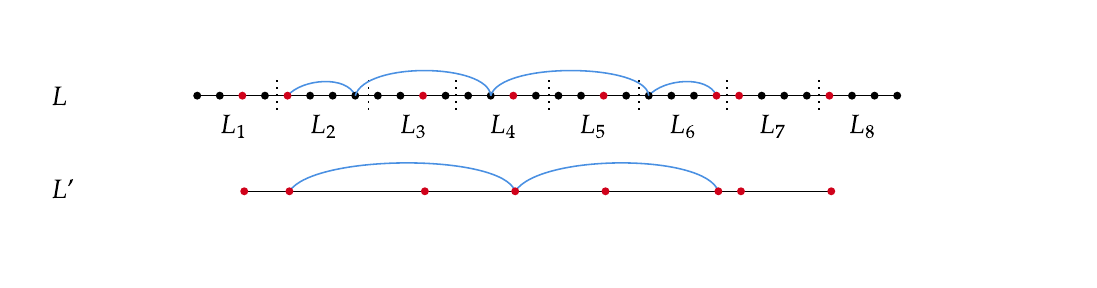}
    \caption{An illustration of Case 1. Every set $L_i$ contains a non-global point. Any $k$-hop path between two non-global points must have the first and the last edge inside the corresponding sets. By contracting every $L_i$ into a single point, we obtain a new path $L'$. Any shortcutting with hop-diameter $k$ for $L$ must have hop-diameter $k-2$ for $L'$. This allows us to use the induction hypothesis for $k-2$. In the illustration above the blue path consisting of 4 hops in $L$ becomes a 2-hop path in $L'$.}
    \label{fig:twCase1}
\end{figure}

\paragraph{Case 2: Some $L_i$ is global.} Let $\{L_l, L_r\} = L\setminus L_i$, so that $L_l$ (resp., $L_r$) is on the left (resp., right) of $L_i$. (Possibly $L_l = \emptyset$ or $L_r = \emptyset$.) At least $|L_i|/2$ points in $L_i$ have edges to either $L_l$ or $L_r$. Without loss of generality, we assume the former. Let $L'_i \subseteq L_i$ be the subset of points that have edges to $L_l$ and 
let $H_i$ be the subgraph of $H$ restricted to preserving distances in $L_i$.
Inductively, $H_i$ has a clique minor of size at least $W_k(n/(2\ell_k))$. (Since $H_i$ is a shortcutting, it does not include any point outside of $L_i$.)
Then ${H}_i$ and ${L}_l$ are vertex-disjoint (because we are considering shortcuttings) and hence their union has a clique minor of size $W_k(n/(2\ell_k))+1$.
Thus, $W_k(n)$ satisfies \cref{eq:lb-rec-k}, which we lower bound next.
\begin{figure}[h]
    \centering
    \includegraphics[width=\linewidth]{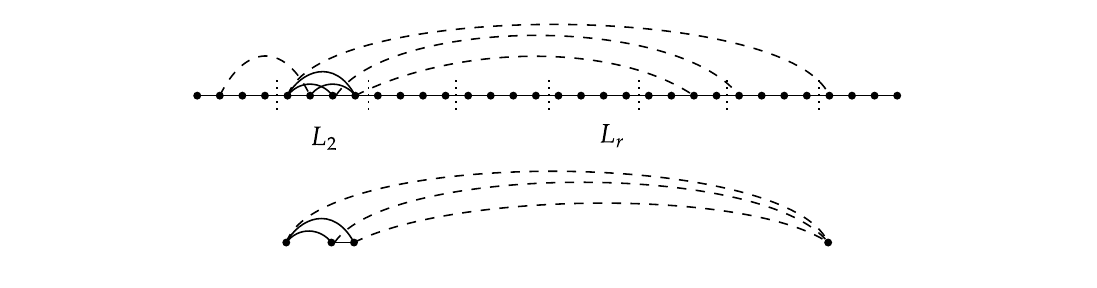}
    \caption{An illustration of Case 2. $L_2$ is global, meaning that every vertex is adjacent to an edge going outside of $L_2$. At least half of these edges are adjacent on the vertices in $L_r$. By contracting $L_r$ and using the clique minor of $L_2$, we obtain the desired clique minor.}
    \label{fig:twCase2}
\end{figure}
\begin{align*}
W_k(n)&\ge \min(c_1k(\log{n})^{2/k}-1,W_k({n}/(2\ell_k))+1) \\
&\ge \min(c_1k(\log{n})^{2/k}-1, c_1k(\log{n} - \log{\ell_k}-1)^{2/k})\\
& \ge \min\left(c_1k(\log{n})^{2/k}-1, c_1k\left(\log{n} - \left(\frac{k}{k-2}\right)^{(k-2)/2}\cdot (\log{n})^{(k-2)/k}-1\right)^{2/k}  \right)\\
& \ge \min\left(c_1k(\log{n})^{2/k}-1, c_1k(\log{n})^{2/k}-1\right)
\end{align*}
The last inequality follows from \Cref{eq:identity} by replacing $x = \log{n}$ and $\alpha=1$.
\end{proof}

\begin{lemma}\label{lem:tw-lb-large-k}
For every $n \ge 1$, every $k > \frac{2}{\ln(2e)}\ln\log{n}$, every 1-shortcutting with hop-diameter $k$ for an $n$-vertex path has treewidth at least $c_2(\log\log{n})^2/k$, for an absolute constant $c_2$.
\end{lemma}
\begin{proof}
We set $\ell_k$ so that $\log\log\ell_k = \sqrt{\frac{k-2}{k}}\log\log{n}$. (We use $\log(\cdot) \coloneqq \log_2(\cdot)$.) For the clarity of exposition, we ignore the rounding issues. We note that $1 \le \ell_k \le n$. 
Using the same argument as in \Cref{lem:tw-lb-small-k}, we have 
$W_k(n) \ge \min(W_{k-2}(\ell_k), W_k(n/(2\ell_k))+1)$. We prove the lemma by induction, where the base case is \Cref{lem:tw-lb-small-k} whenever $k > \frac{2}{\ln(2e)}\ln\log{n}$.
Our goal is to prove that $W_{k-2}(\ell_k) \ge c_2(\log\log{n})^2/k$ and $W_k(n/(2\ell_k))+1 \ge c_2(\log\log{n})^2/k$. 
For the first inequality we distinguish two cases. If $k-2 \le \frac{2}{\ln(2e)} \ln \log(\ell_k)$, then by \Cref{lem:tw-lb-small-k} we have 
\begin{align*}
W_{k-2}(\ell_k) 
&\ge c_1(k-2) \log^{\frac{2}{k-2}}\ell_k - 1 
\ge c_1(k-2) \log^{\frac{\ln(2e)}{\ln\log{\ell_k}}}\ell_k-1 
= 2ec_1(k-2)-1
\ge ec_1k-1\\
&\ge ec_1 \cdot \left(\frac{2}{\ln(2e)}\right)^2\frac{(\ln\log{\ell_k})^2}{k}-1
\ge c_2\frac{(\log\log{\ell_k})^2}{k-2}
\end{align*}
The penultimate inequality holds because $k \ge \frac{2}{\ln(2e)} \ln \log(n) \ge \frac{2}{\ln(2e)} \ln \log(\ell_k)$.
The last inequality holds for a proper choice of constant $c_2$. If $k-2 > \frac{2}{\ln(2e)} \ln \log(\ell_k)$, we have $W_{k-2}(\ell_k) \ge c_2\frac{(\log\log{\ell_k})^2}{k-2}$ by the induction hypothesis. Hence, in both cases, we have:
\begin{align*}
W_{k-2}(\ell_k) &\ge c_2 \cdot  \frac{(\log\log{\ell_k})^2}{k-2} = c_2 \cdot \frac{\left(\sqrt{\frac{k-2}{k}}\log\log{n} \right)^2}{k-2} = c_2 \cdot \frac{(\log\log{n})^2}{k}
\end{align*}
For the second inequality, let $x=\log{n}$. We have $\log\ell_k = x^{\sqrt{(k-2)/k}}$. Since $\frac{n}{2\ell_k} < n$, we have that $k > \frac{2}{\ln(2e)}\ln\log{\frac{n}{2\ell_k}}$ and the induction hypothesis gives the following.
\begin{align*}
W_k\left(\frac{n}{2\ell_k}\right) + 1 \ge \frac{c_2\left(\log\log\left(\frac{n}{2\ell_k}\right)\right)^2}{k}+1 = \frac{c_2\log^2\left(x-x^{\sqrt{(k-2)/k}}-1\right)}{k}+1
\end{align*}
To show that the right-hand side is at least $c_2\log^2(x)/k$, it suffices to show the following: \begin{equation}
\log^2\left(x-x^{\sqrt{(k-2)/k}}-1\right) + \frac{k}{c_2}\ge \log^2{x}\label{eq:goal}
\end{equation}
From \Cref{clm:taylor} we have $x^{\sqrt{(k-2)/k}} \le x - \frac{x\ln{x}}{k}+\frac{x\ln^2{x}}{k^2}$.
\begin{align*}
x-x^{\sqrt{(k-2)/k}}-1 &\ge x-\left(x - \frac{x\ln{x}}{k}+\frac{x\ln^2{x}}{k^2} \right)-1  \\
&= \frac{x\ln{x}}{k}\left(1-\frac{\ln{x}}{k} \right)-1 \\
&> \frac{x\ln{x}}{10k}-1  \\
\end{align*}
The last inequality holds because $k > \frac{2}{\ln(2e)}\ln(x)$. We next consider two cases. If $\frac{x\ln{x}}{10k} < 2$,
 then $\frac{k}{c_2} > \frac{x\ln{x}}{20c_2} \ge \log^2{x}$ and \Cref{eq:goal} is proved. Otherwise, we proceed as follows.
\begin{align*}
\log^2\left(\frac{x\ln{x}}{10k}-1 \right) + \frac{k}{c_2} & \ge \log^2\left(\frac{x\ln{x}}{20k} \right) + \frac{k}{c_2}\\
&=\left(\log{x} + \log\frac{\ln{x}}{20k}\right)^2 + \frac{k}{c_2}\\
&= \log^2{x} + 2(\log{x})\log\frac{\ln{x}}{20k} + \log^2\frac{\ln{x}}{20k} + \frac{k}{c_2}\\
&\ge \log^2{x}
\end{align*}
The last inequality holds for any $c_2 \le 1/10$.
\end{proof}
\section{Low arboricity shortcuttings}\label{sec:arb}

Throughout this section, we use the following lemma.
\begin{lemma}\label{lem:orient}
If every edge of a graph $G=(V,E)$ can be oriented such that the maximum in-degree of every vertex is at most $d$, then the arboricity of $G$ is at most $d+1$.
\end{lemma}

\subsection{Paths}\label{sec:arb-line}

In this section, we show a construction for paths (\Cref{thm:arbUBline}). We shall use a modification of the following well-known result.
\begin{theorem}[Cf. \cite{AS24,BTS94,Sol13}]\label{thm:hop-sparse}
For every $n \ge 2$ and $k\ge 2$, every $n$-point tree admits a shortcutting with hop-diameter $k$ and $O(n\alpha_k(n))$ edges.
\end{theorem}
We next state a slightly modified version of the previous theorem. The first statement concerns hop-diameter $1$.

\begin{lemma}\label{lem:hop-sparsity-double-1}
Let $n \ge 2$ be an arbitrary integer. Let $L$ be a path induced by a set of $n$ points on a line so that between every two points there is $n$ Steiner points. Let $S$ denote the set of Steiner points. Then, $L$ admits a Steiner shortcutting with hop-diameter $2$ such that the Steiner points belong to $S$ and every vertex in $L \cup S$ has a constant in-degree. 
\end{lemma}
\begin{proof}
Interconnect the vertices in $L$ by a clique. Consider an arbitrary clique edge $(u,v)$ and split it into two using a Steiner point $w$. Orient the edges $(u,w)$ and $(w,v)$ into $w$. By using a fresh Steiner point for every clique edge, we obtain the guarantees from the statement.
\end{proof}
Next, we state the general version.

\begin{lemma}\label{lem:hop-sparsity-double}
Let $n \ge 2$ and $k \ge 2$ be two arbitrary integers. Let $L$ be a path induced by a set of $n$ points on a line so that between every two points there is $\alpha_k(n)$ Steiner points. Let $S$ denote the set of Steiner points. Then, $L$ admits a Steiner shortcutting with hop-diameter $2k$ such that the Steiner points belong to $S$ and every vertex in $L \cup S$ has a constant in-degree. 
\end{lemma}
\begin{proof}
We prove the lemma by induction over $k$. For $k=2$, we take a central vertex $c$ in $L$ and connect it to every other point in $L$; orient the edges outwards from $c$. Proceed recursively with the two halves. This way we obtain a shortcutting for $L$ with hop-diameter $2$. Denote by $E$ the edge set of this shortcutting. The depth of the recursion in the construction is $O(\log{n})$, and the size of $S$ is $n\alpha_2(n) = n\log{n}$. Every edge in $E$ has in-degree 1 for every recursion level. We can split every such edge into two using a fresh Steiner point for each recursion level. This concludes the proof for $k=2$.

Consider now an arbitrary $k$. Divide $L$ into intervals of size $\alpha_{k-2}(n)$ using $n/\alpha_{k-2}(n)$ cut vertices. Denote by $C$ the set of cut vertices and invoke the induction hypothesis on $C$ with parameter $k-2$. Let $E'$ be the obtained set of edges. Let $E_C$ be obtained by connecting every cut vertex to every point in the two neighboring intervals. Let this edge set be $E_C$. Proceed recursively with parameter $k$ on each of the intervals. 

To analyze the in-degree, we observe that the depth of the recursion with parameter $k$ is $O(\alpha_k(n))$, which coincides with the number of Steiner vertices between every two points in $L$. One level of recursion contributes a constant in-degree to each vertex in the construction. This means that we can split all such vertices into two and use a fresh Steiner point at each recursion level. This concludes the proof.
\end{proof}

 We next prove \Cref{thm:arbUBline}.
 
\begin{restatable}{theorem}{arbUBline}\label{thm:arbUBline}
For every  $n \ge 1$ and every even $k \ge 2$, every $n$-point path admits a shortcutting with hop-diameter $k$ and arboricity $O(\alpha_{k/2+1}(n))$.
\end{restatable}
\begin{figure}
\centering
\includegraphics[width=\linewidth]{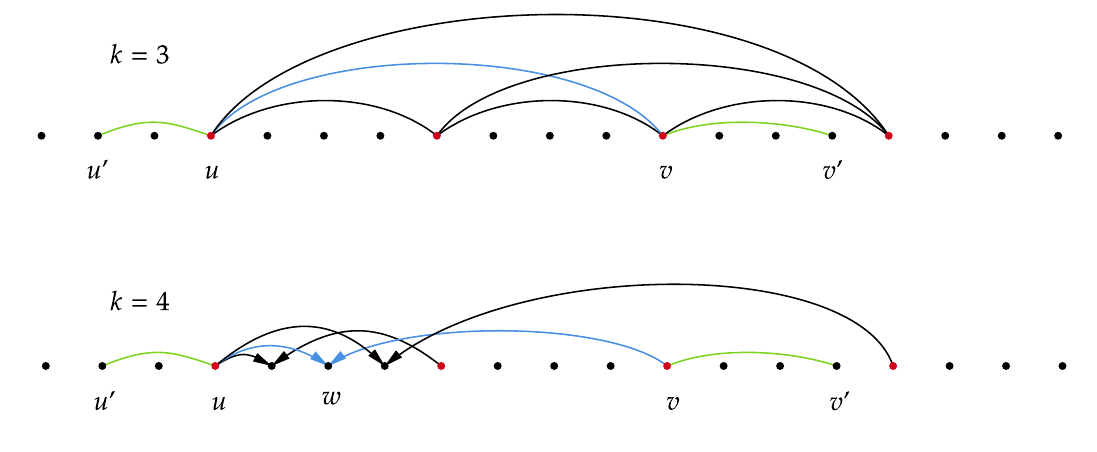}
\caption{An illustration of the shortcutting with hop-diameter 4 and arboricity $O(\log\log{n})$. We start by a well-known 3-hop construction \cite{AS87,BTS94}. Split the path into subpaths using $\sqrt{n}$ cut vertices (red) so that every subpath contains $\sqrt{n}$ vertices. The shortcutting is obtained by: \textit{(i)} interconnecting the cut vertices by a clique, \textit{(ii)} connecting every cut vertex to the points in its two adjacent subpaths (in the illustration, only two such edges, $(u',u)$ and $(v,v')$, are shown in green), and \textit{(iii)} applying the recursive construction to each of the subpaths. This construction readily gives hop-diameter 3 (see, e.g., the path between $u'$ and $v'$), but it suffers from having arboricity $\Omega(\sqrt{n})$ due to the clique on the cut vertices. To remedy this, we split every clique edge $(u,v)$ into two edges, $(u,w)$ and $(w,v)$. This allows us to reduce the in-degree of every point to $O(1)$ per recursion level, at a cost of increasing the hop-diameter to 4. Overall, the number of recursion levels is $O(\log\log{n})$, yielding the desired bound on the arboricity.
}
\label{fig:arbLine}
\end{figure}
\begin{proof}
Let $L_n$ be an arbitrary path. For an integer $k'\ge 2$, we describe a construction of a shortcutting $H$ for $L_n$ with hop-diameter $4k'-2$ and arboricity $O(\alpha_{2k'}(n))$.

Consider a set of $n'=n/\alpha_{2k'-2}(n)$ equally-spaced \emph{cut vertices} dividing $L_n$ into intervals of size $\alpha_{2k'-2}(n)$. To construct the shortcutting $H$, we connect every cut vertex to all the vertices in its two neighboring intervals. Denote the corresponding edge by $E_C$.
Let $C$ be the set of cut vertices. Let $E'$ be obtained by invoking \Cref{lem:hop-sparsity-double} with parameter $2k'-2$ on $C$ using $L_n$ as Steiner points. Proceed recursively with each of the intervals.

To analyze the arboricity, we will show that the edges in $E_C$ and $E'$ can be oriented so that the in-degree of every vertex is constant. Orient every edge in $E_C$ so that it goes out of the corresponding cut vertex. Since every interval is adjacent to at most two cut vertices, the in-degree of every point with respect to $E_C$ is at most 2. By \Cref{lem:hop-sparsity-double}, the edges in $E'$ have a constant in-degree on $L_n$. In conclusion, $E_C$ and $E'$ contribute $O(1)$ to the in-degree of each vertex in $L_n$. The number of recursion levels $\ell(n)$ satisfies the recurrence $\ell(n) = \ell(\alpha_{2k'-2}(n))+O(1)$, which has solution $\ell(n) = \alpha_{2k'}(n)$. Every recursion level contributes $O(1)$ to the in-degree of vertices, meaning that the overall in-degree in $H$ is $O(\alpha_{2k'}(n))$. The hop-diameter of $H$ is $2 + 2\cdot(2k'-2)=4k'-2$. This concludes the description of a shortcutting with hop-diameter $4k'-2$ and arboricity $O(\alpha_{2k'}(n)$.

Note that we have only shown how to get a tradeoff of hop-diameter $4k'-2$ and arboricity $O(\alpha_{2k'}(n))$. We could similarly get a construction with hop-diameter $4k'$ and arboricity $O(\alpha_{2k'+1}(n))$ for a parameter $k'\ge 1$. Specifically, we divide $L_n$ into intervals of size $\alpha_{2k'-1}(n)$. The cut vertices are interconnected using \Cref{lem:hop-sparsity-double} with hop-diameter $2k'-1$ (and \Cref{lem:hop-sparsity-double-1} when $k'=1$). The hop-diameter is $2 + 2(2k'-1)=4k'$ and the arboricity is $O(\alpha_{2k'+1}(n))$ using a similar argument.
\end{proof}

The construction for ultrametric is a simple adaptation of the construction for paths, which we will show next.  

\subsection{Ultrametrics and doubling metrics}

A metric $(X,d_X)$ is an \emph{ultrametric} if it satisfies a strong form of the triangle inequality: for every $x,y,z$, $d_X(x,z)\leq \max\{d_X(x,y), d_X(y,z)\}$. It is well known that an ultrametric can be represented as a \emph{hierarchical well-separated tree} (HST). 

More precisely, an \emph{$(s,\Delta)$-HST} is a tree $T$ where (i) each node $v$ is associated with a label $\Gamma_v$   such that $\Gamma_v \geq s \cdot \Gamma_u$ whenever $u$ is a child of $v$ and (iii) each internal node $v$ has at most $\Delta$ children. Parameter $s$ is called the \emph{separation} while $\Delta$ is called the degree of the HST of the HST.  Let $L$ be the set of leaves of $T$. The labels of internal nodes in $T$ induce a metric $(L,d_L)$ on the leaves, called leaf-metric, where for every two leaves $u,v \in L$, $d_L(u,v) = \Gamma_{\lca(u,v)}$ where $\lca$ is the lowest common ancestor of $u$ and $v$. It is well-known, e.g., \cite{Bartal2004}, that  $(L,d_L)$ is an ultrametric, and that any ultrametric is isomorphic to the leaf-metric of an HST.

Chan and Gupta~\cite{CG06} showed that any $(1/\eps, 2)$-HST can be embedded into the line metric with (worst-case) distortion $1+O(\eps)$. Therefore, by applying \Cref{thm:arbUBline}, we obtain a $(1+O(\eps))$-spanner with hop diameter $k$ and arboricity $O(\alpha_{k/2+1}(n))$ for $(1/\eps, 2)$-HST. In our setting, we are interested in large-degree $(k,\Delta)$-HST where $s = 1/\eps$ and $\Delta = \poly(1/\eps)$; the embedding result by  Chan and Gupta~\cite{CG06} no longer holds for these HSTs. Instead, we directly apply our technique for the line metric to get a $(1+O(\eps))$-spanner with low-hop diameter.

\begin{lemma}\label{lm:kd-HST-shortcut} Let $\eps \in (0,1), \Delta > 0$ be parameters, and $k$ be an even positive integer. Then, any $(1/\eps, \Delta)$-HST with $n$ leaves admits a $(1+O(\eps))$-spanner with hop-diameter $k$ and arboricity $O(\alpha_{k/2+1}(n))$.
\end{lemma}
\begin{proof}  
Let $T$ be the $(1/\eps, \Delta)$-HST and let $M_T$ be the metric induced by $T$.
For an integer $k' \ge 2$, we describe a construction of a $(1+O(\eps))$-spanner for $M_T$  with hop-diameter $4k'-2$ and arboricity $O(\alpha_{2k'}(n))$. The construction is similar to that in \Cref{thm:arbUBline}.

Let $C$ be the set of internal nodes of $T$, called \emph{cut verices}, such that the subtrees rooted at these nodes has size $\alpha_{2k'-2}(n)$.  The number of cut vertices is $|C| \leq n/\alpha_{2k'-2}(n)$. 
First, connect every cut vertex to all of its descendants in $T$ and let the corresponding set of edges be $E_C$. Next, let $E$ be the set of edges interconnecting the cut vertices using \Cref{thm:hop-sparse} with hop-diameter $2k'-2$ and $O(n'\alpha_{2k'-2}(n'))=O(n)$ edges. We construct a set $E'$ by subdividing every edge $(u,v) \in E$ into two edges using a vertex, say $x$, in the subtree rooted at $u$.  The spanner $H$ is obtained using the edges in $E_C$ and those in $E'$. Finally, the recursive construction is applied to each subtree rooted at a vertex in $C$. This concludes the description of the recursive construction of $H$. 

The same argument in \Cref{thm:arbUBline} implies that the arboricity is $O(\alpha_{2k'}(n))$. The stretch is  $(1+O(\eps))$ since the path $u\rightarrow x\rightarrow v$  between two cut vertices $u$ and $v$ has stretch $(1+O(\eps))$.
\end{proof}

To construct a low-hop spanner with small arboricity for doubling metrics (\Cref{thm:arbUBdoubling}), we will rely on the ultrametric cover by Filtser and Le~\cite{FL22lso}. Following their notation, for a given metric space $(X,d_X)$, we say that a collection $\mathcal{T}$ of at most $\tau$ different $(s,\Delta)$-HSTs such that (i) for every  HST $T\in \mathcal{T}$, points in $X$ are leaves in $T$, and (ii) for every two points $x,y\in X$, $d_X(x,y)\leq d_T(x,y)$ for every  $T\in \mathcal{T}$, and there exists a tree $T_{xy}\in \mathcal{T}$ such that $d_X(x,y)\leq \rho \cdot d_{T_{xy}}(x,y)$.

\begin{theorem}[Cf. Theorem 3.4 in \cite{FL22lso}]\label{thm:embed}
For every $\eps \in (0,1/6)$, every metric with doubling dimension $d$ admits an $(\eps^{-O(d)},1+O(\eps),1/\eps, \eps^{-O(d)})$-ultrametric cover. 
\end{theorem}

\arbUBdoubling*
\begin{proof} 
Let $\mathcal{T}$ be the $(\eps^{-O(d)},1+O(\eps),1/\eps, \eps^{-O(d)})$-ultrametric cover in \Cref{thm:embed} for the input doubling metric.  The theorem then follows by applying \Cref{lm:kd-HST-shortcut} to each of $(1/\eps,\eps^{-O(d)})$-HST in $\mathcal{T}$ and taking the union of the resulting spanners. 
\end{proof}

\subsection{General tree metrics}\label{sec:arb-tree}
\begin{restatable}{theorem}{arbUBtree}\label{thm:arbUBtree}
For every two integers $n \ge 1$ and $k \ge 1$ and every $n$-vertex tree $T$, there is a shortcutting with hop-diameter $k$ and arboricity $O(\log^{12/(k+4)}{n})$. Moreover, when the height of the tree is $h$, then the arboricity is $O(h^{6/(k+4)})$.
\end{restatable}
\begin{proof}[Proof of \Cref{thm:arbUBtree} for height $h$]
We first show how to prove the theorem for a tree $T$ with height bounded by $h$. This construction gives the main ideas used for general trees. Let $k'$ be an arbitrary integer. We show a recursive construction with hop-diameter $2k'$. In particular, we show how to shortcut $T$ so that between every ancestor and descendant it is possible to go using $k'$ hops. The construction is the same as in \Cref{lem:treewidth-2}: take the root of the tree, connect it to every descendant and proceed recursively with each of its children. By orienting the edges from the roots to descendants, we have that the in-degree of every vertex is bounded by $h$. Let $A_1(h)$ denote the obtained in-degree (and thus arboricity) of the shortcutting. We have that $A_1(h) = h$.

We next show the bound for an arbitrary $k' = 1 + 3g$ for an integer $g \ge 1$.
Let $\ell$ be a parameter to be set later. Consider the tree levels so that the root is at level $0$.
Designate as the \emph{cut vertices} all the vertices at levels $\ell, 2\ell,3\ell\ldots$. Denote the set of cut vertices by $C$. Consider the set $S$, consisting of all the parents of vertices in $C$. Let $E_{CS}$ be obtained by interconnecting all the vertices in $C$ to their parents. Each such edge is oriented from a vertex in $S$ towards the vertex in $C$. Next, connect every vertex in $S$ to its first $\ell-1$ ancestors until the occurrence of the first cut vertex; let the corresponding edge set be $E_S$. Every such edge is oriented from the ancestors towards vertices in $S$. Let $c \in C$ be an arbitrary vertex at level $d$. Connect $c$ to all of its descendants at levels $d + 1, d +2,\ldots, d+\ell-2$, i.e., until the nest cut vertex, and orient these edges from $c$ towards the descendants. The edge set $E_C$ is obtained by doing this for every vertex $c$ in $C$. Use a recursive construction with parameter $k'-3$ on the subtree of $T$ induced by vertices in $C$. Finally, consider all the subtrees obtained by removing $C$ and $S$ from  $T$ and apply the recursive construction with parameter $k'$ on each of the subtrees. 

We next analyze the hop-diameter of the ancestor-descendant paths in this construction. Let $u$ and $v$ be two arbitrary vertices such that $v$ is an ancestor of $u$. The path between $u$ and $v$ in the shortcutting is as follows. By construction, $E_C$ contains an edge between $u$ and its ancestor cut vertex $c_u \in C$. Let $d_u$ be the highest cut vertex that is ancestor of $c_u$ and a descendant of $v$. There is a path consisting of at most $k'-3$ hops between $c_u$ and $d_u$. Let $s_u \in S$ be the parent of $d_u$. The edge $(d_u,s_u)$ is in $E_{CS}$. Finally, $E_S$ contains an edge $(s_u,v)$. Clearly, the path consists of $k'$ hops.

We next analyze the in-degree of vertices in $T$. Let $A_{k'}(h)$ denote the in-degree of the construction with parameter $k'$. Then, $A_{k'}(h) = \ell-1$ for the vertices in $S$, due to the orientation of the edges in $E_S$. For the vertices in $C$, we have $A_{k'}(h) = 1+A_{k'-3}(h/\ell)$, because the edges in $E_{CS}$  add one to in-degree of every vertex and the dominant term is due to the recursive call with parameter $k'-3$. Finally, all the other vertices have $A_{k'}(h) = 1+A_k(\ell-2)$, where the edges in $E_C$ contribute one to the in-degree and $A_k(\ell)$ is due to the recursive construction with parameter $k'$.
Putting everything together, we have the following recurrence .
\begin{equation}\label{eq:rec-arb}
A_{k'}(h) = \max(\ell-1, 1+A_{k'-3}(h/\ell), 1+A_{k'}(\ell-2))
\end{equation}
We proceed to show inductively that for every $k'=1+3g$ we have $A_{k'}(h) \le h^{1/(g+1)}$. 
Since we have that $A_{k'}(h) \le h$ for every $k' \ge 1$, we can disregard the third term in \Cref{eq:rec-arb}. Thus, we obtain the following simplified recurrence: $A_{k'}(h) = \max(\ell, A_{k'-3}(h/\ell))$. Notice that we have replaced $\ell-1$ by $\ell$ in the first term since it does not affect the solution asymptotically. We proceed to solve the recurrence.
\begin{align*}
A_{k'}(h) &\le \max(\ell, A_{k'-3}(h/\ell))\\
&\le \max(\ell, (h/\ell)^{1/g}) & \text{induction hypothesis}\\
&\le \max(h^{1/(g+1)}, (h^{1-1/(g+1)})^{1/g}) & \text{setting } \ell = h^{1/(g+1)}\\
&\le \max(h^{1/(g+1)}, h^{1/(g+1)})\\
&=h^{1/(g+1)}
\end{align*}

Thus $A_{k'}(h) \le h^{1/(g+1)}$. In particular, we have that $g = (k-1)/3$ so the tradeoff is $k'$ versus arboricity $h^{3/(k'+2)}$. To get the tradeoff guaranteed in the statement, we observe that the hop-diameter is $2k'$.

\end{proof}
\begin{proof}[Proof of \Cref{thm:arbUBtree} for general trees]
We consider a heavy-light decomposition of $T$, constructed as follows. Start from the root $r$ down the tree each time following the child with the largest subtree. The obtained path is called the heavy path rooted at $r$. Continue recursively with every child of the vertices of the heavy path. Let $T'$ be obtained by contracting every heavy path in $T$ into a single vertex. It is well-known that the height of $T'$ is $\log{n}$, where $n$ is the number of vertices in $T$. 

We start by employing the shortcutting procedure for bounded height trees on $T'$ with parameter $k'$ and explain how to adapt it to $T$. Consider an arbitrary edge $(u,v)$ in the shortcutting of $T$ such that $u$ is a descendant of $v$. Let $P_u$ (resp., $P_v$) be the heavy path in $T$ corresponding to $u$ (resp., $v$). Add an edge between the root $r_u$ of $P_u$ and its lowest ancestor on $P_v$. We do this for all the edges in the shortcutting of $T'$. For every heavy path $P$ in $T$, use the 4-hop construction with arboricity $O(\log\log{n})$ from \Cref{thm:arbUBline}. In addition, connect via a direct edge every vertex in $p$ to the root of $P$. This concludes the description of the shortcutting for $T$. 

Next, we analyze the hop-diameter of the obtained construction. Let $u$ and $v$ be two arbitrary vertices in $T$ and let $P_u$ and $P_v$ be the corresponding heavy paths in $T$. Denote by $p_u$ and $p_v$ the vertices in $T'$ corresponding to $P_u$ and $P_v$. Let $p_w$ be the LCA of $p_u$ and $p_v$ in $T'$. From the construction we know that there is a $k'$-hop path between $p_u$ and $p_w$ and between $p_v$ and $p_w$. Let $(p_v,p_a)$ be the first edge on the path from $p_u$ to $p_w$. The corresponding path in $T$ goes from $u$ to the parent of $P_u$ and from the root of $P_u$ to its lowest ancestor in $P_a$. We can similarly replace every edge on the path between $P_u$ and $P_w$ in $T'$ by two edges in $T$. We handle analogously the path between $P_v$ and $P_w$. The corresponding paths in $T$ go from $u$ to its lowest ancestor on $P_w$ and from $v$ to its lowest ancestor on $P_w$. Using the edges from the 4-hop construction on $P_w$, we join the two paths. The overall number of hops is $4k'+4$. In particular, we achieve a hop-diameter of $4k'+4$ with arboricity of $O(\log^{3/(k'+2)}{n})$. Letting $k = 4k'+4$, the arboricity is $O(\log^{12/(k+4)}{n})$.
\end{proof}

Using \Cref{thm:arbUBtree} and the known tree cover constructions, we obtain low arboricity spanners for various metric spaces. See \Cref{tab:arb}.

\begin{table}
\centering
\begin{tabular}{|l|l|l|}
\hline
metric     & stretch  & arboricity \\\hline
planar \cite{CCL+23FOCS} & $1+\eps$ & $O(\epsilon^{-3}\log(1/\epsilon)\log^{12/(k+4)}{n})$\\ \hline
minor-free \cite{CCL+24SODA} & $1+\eps$ &  $2^{r^{O(r)}/\eps}\cdot \log^{12/(k+4)}{n}$ \\\hline
general \cite{MN07} & $O(\gamma)$ & $O(\gamma \cdot n^{1/\gamma} \cdot \log^{12/(k+4)}{n})$\\\hline
\end{tabular}
\caption{Low-arboricity spanners for various metric spaces using \Cref{thm:arbUBtree}.}\label{tab:arb}
\end{table}
\section{Routing upper bounds}\label{sec:routing-ub}

In this section, we show the results for routing in tree metrics and doubling metrics. 
\subsection{Routing in tree metrics}
\begin{restatable}{theorem}{routingUB}\label{thm:routingUB}
For every $n$ and every $n$-vertex tree $T$, there is a 3-hop routing scheme in the fixed-port model for the metric $M_T$ induced by $T$ with stretch 1 that uses $O(\log^2{n}/\log\log{n})$ bits per vertex.
\end{restatable}
\begin{proof}
Our routing scheme is constructed on top of a $1$-spanner $H_3$ of $M_T$ as described in \Cref{lem:treewidth-3}. For a vertex $u \in T$, let $\tbl(u)$ denote its routing table and $\lbl(u)$ its label. First, assign a unique identifier $ID(u) \in \{1, \ldots, n\}$ to every vertex $u$ in $T$ and add it to $\tbl(u)$ and $\lbl(u)$.
Equip the routing table and a label of every vertex $u \in T$ with an ancestor label $\anc(u)$ as in \cite{AAKMR06}. This adds $O(\log{n})$ bits of memory per vertex. Using the ancestor labeling scheme from \cite{AAKMR06}, we can determine, given two vertices $u$ and $v$, whether they are in ancestor-descendant relationship, and if so, whether $u$ is the ancestor of $v$ or the vice-versa.

Recall the recursive construction of $H_3$ with parameter $T$ as an input. Assign to each recursive call a unique integer $r_T$. Let $X$ be a set of vertices for $T$ as in \Cref{tree-partition} so that $|X| =\log{n}/\log\log{n}$. (This parameter is fixed and does not change across different recursive calls.) The vertices of $X$ are interconnected by a clique in $H_3$. For every vertex in $u \in X$, add to $\tbl(u)$ the information consisting of $C(u) = \lAngle{r_T, \{\lAngle{ID(v), \port(u,v), \anc(v)} \mid v \in X \setminus \{x \} \}}$. The memory occupied per every vertex in $X$ is $O(\log^2{n}/\log\log{n})$. (Note that the construction of $H_3$ guarantees that every vertex belongs to such a clique exactly once across all the recursive calls, meaning that $\tbl(u)$ contains only one such $C(u)$.)
Let $T'$ be a subtree in $T \setminus X$. Let $u$ and $v$ be two vertices from $X$ to which $T'$ has outgoing edges. For every vertex $x \in T'$, add to $\tbl(x)$ the following: $\lAngle{r_T, ID(u), \port(x,u)}$ and $\lAngle{r_T, ID(v), \port(x,v)}$. Similarly, add to $\lbl(x)$ the following: $\lAngle{r_T, ID(u), \port(u,x)}$ and $\lAngle{r_T, ID(v), \port(v,x)}$. This information takes $(\log{n})$ bits per recursive call $r_T$.
The construction proceeds recursively with $T'$; the number of recursive calls every vertex participates in is at most $O(\log{n}/\log\log{n})$.

Next we describe the routing algorithm. Let $u$ be the source and $v$ be the destination. First, check if $C(u)$ contains routing information leading directly to $v$. In this case, the algorithm outputs $\port(u,v)$ and the routing is complete. (This case happens when $u$ and $v$ are in the same clique during the construction.) Otherwise, go over $\tbl(u)$ and $\lbl(v)$ and find the last recursive call $r_T$ which is common to both $u$ and $v$. Next, consider $\lbl(v)$ and the two entries consisting $\lAngle{r_T, ID(v_1), \port(v_1,v)}$ and $\lAngle{r_T, ID(v_2), \port(v_2,v)}$, corresponding to $r_T$. If $v_1$ and $v_2$ are in $C(u)$, use $\anc(u)$, $\anc(v_1)$, $\anc(v_2)$, and $\anc(v)$ to decide whether to output $\port(u,v_1)$ or $\port(u,v_2)$.  (This case happens when $u$, $v_1$, and $v_2$ are in $X$ in the recursive call $r_T$ and $v$ is not in it.) Finally, let $\lAngle{r_T, ID(u_1), \port(u,u_1)}$ and $\lAngle{r_T, ID(u_2), \port(u,u_2)}$ 
be the two entries corresponding to $r_T$ in $\tbl(u)$. Use $\anc(u)$ and $\anc(v)$ to decide whether to output $\port(u,u_1)$ or $\port(u,u_2)$. (This case happens when $u$ is not in $X$.) This concludes the description of the routing algorithm. 
\end{proof}

\subsection{Routing in doubling metrics}
We next show how to extend the routing result in tree metrics to metrics with doubling dimension $d$. In particular, we prove the following theorem.
\routingUBdoubling*
Given a point set $P$ with doubling dimension $d$, we first construct a tree cover, using the tree cover theorem from \cite{CCLST25}.
\begin{theorem}[\cite{CCLST25}]\label{thm:cover-doublign}
Given a point set $P$ in a metric of constant doubling dimension $d$ and any parameter $\eps \in (0,1)$, there exists a tree cover with stretch $(1+\eps)$ and $\eps^{-\tilde{O}(d)}$ trees. Furthermore, every tree in the tree cover has maximum degree bounded by $\eps^{-O(d)}$.
\end{theorem}
We use this specific tree cover theorem, since the authors also provide an algorithm for determining the ``distance-preserving tree'' given the labels of any two metric points.
\begin{lemma}[\cite{CCLST25}]\label{lem:find-tree}
Let $\eps \in (0,1)$. Let $T = \{T_1, \ldots, T_k\}$ be the tree cover for $P$ constructed by \Cref{thm:cover-doublign}, where $k = \eps^{-\tilde{O}(d)}$. There is a way to assign $\eps^{-\tilde{O}(d)}\log{n}$-bit labels to each point in $P$ so that, given the labels of two vertices $x$ and $y$, we can identify an index $i$ such that tree $T_i$ is a ``distance-approximating tree'' for $u$ and $v$; that is, $\delta_{Ti}(x,y) \le (1+\eps)\delta_P(x,y)$. This decoding can be
done in $O(d \cdot \log( 1/\eps))$ time.
\end{lemma}
We equip each tree in the cover with the stretch-1 routing scheme from \Cref{thm:routingUB}. This consumes overall $\eps^{-\tilde{O}(d)}\log^2n/\log\log{n}$ bits per point in $P$. In addition, we add $\eps^{-\tilde{O}(d)}\log{n}$-bit labels to each point in $P$ as stated in \Cref{lem:find-tree}. Given two points $x,y \in P$, we first employ the algorithm from \Cref{lem:find-tree} to find the tree in which the routing should proceed. Then, the routing proceeds on that specific tree using the routing algorithm from \Cref{thm:routingUB}. This concludes the description of the compact routing scheme for doubling metrics.

\section{Routing lower bound}\label{sec:routing-lb}

In this section, we prove  the following theorem.
\begin{restatable}{theorem}{routingLB}\label{thm:routingLB}
There is an infinite family of trees $T_n$, $n>0$, such that any labeled fixed-port routing scheme with stretch 1 on a metric induced by $T_n$ has at least one vertex with total memory of $\Omega(\log^2{n}/\log\log{n})$ bits.
\end{restatable}

Let $T$ be an unweighted tree and $M_T = (V, (V \times V), d_T)$ be a metric induced by $T = (V,E)$. 
The edges in $(V \times V) \setminus E$ are called \emph{Steiner edges}. In this section we show that stretch-$1$ routing in tree metrics requires $\Omega(\log^2{n}/\log\log{n})$ bits per tree vertex. 

\paragraph{Hard instances.} We first describe the hard instances used in \cite{FG02}. Let $t$, $h$, and $d$ be positive integers and let $\overline{T}_0$ be a complete $t$-ary rooted tree of height $h+1$. Let $T_0$ be a tree obtained by adding $d-t-1$ leaves at every internal vertex of $\overline{T}_0$ and $d-t$ leaves at the root. These added leaves are called \emph{dummy leaves}. The number of vertices in $T_0$ is $n= (d-1) \cdot \frac{t^h - 1}{t-1} + 2$. Note that $T_0$ is uniquely defined by $t$, $h$, and $d$.

Let $T$ be a tree obtained from $\overline{T}_0$ as follows. Consider an internal vertex $u$ of $\overline{T}_0$ at height $i$, where the root has height $h$ and the leaves have height $0$. 
Let $q_i = \frac{t^{i} - 1}{t-1}$. (The choice of $q_i$ coincides with the number of non-dummy vertices in a subtree rooted at any child of $u$.) Add $(d-t-1)\cdot q_i$ dummy leaves to $u$ if it is an internal node and $(d-t)q_i$ if it is the root. Note that both $T_0$ and $T$ are constructed based on $\overline{T}_0$ and there is a correspondence between non-dummy vertices of $T$ and the non-dummy vertices of $T_0$. In what follows, we shall use the same letter to denote some non-dummy vertex in $T$ and the corresponding non-dummy vertex in $T_0$.

\begin{claim}
The number of vertices in $T$  is  $O(n\log^2{n})$.
\end{claim}
\begin{proof}
\begin{align*}
|T| &= |\overline{T}_0| + q_h + \sum_{i=1}^h t^{h-i}\cdot (d-t-1) \cdot q_i\\
&= \frac{t^{h+1}-1}{t-1} + \frac{t^{h}-1}{t-1} + \frac{d-t-1}{t-1}\cdot \sum_{i=1}^h t^{h-i}\cdot (t^i-1)\\
&=\frac{t^{h+1}-1}{t-1} + \frac{t^{h}-1}{t-1} + \frac{(d-t-1)(ht^{h+1}-ht^h-t^h+1)}{(t-1)^2}
\end{align*}
In \cite{FG02}, $d=t^h$ and $t=h=\floor{(\log\sqrt{n})/\log\log\sqrt{n}}$, so that $d=t^h \le \sqrt{n}$. We proceed to upper bound $|T|$ as follows.
\begin{align*}
|T| \le 2t^{h+1} + dht^{h+1} \le 2\sqrt{n}\log{n}+n\log^2{n} = O(n\log^2{n})
\end{align*}
\end{proof}
Using a reduction to the lower bound instances of \cite{FG02}, we will show that the memory requirement is $\Omega(\log^2n'/\log\log{n'}) = \Omega(\log^2n/\log\log{n})$.

\paragraph{Reduction to relaxed routing.} Let $M_T$ be a metric induced by $T$. Similarly to \cite{FG02}, we consider a restricted problem of \emph{relaxed routing} in $M_T$, where the destination vertex is a non-dummy vertex of $M_T$ and the source vertex is its ancestor. Our lower bound argument shows that relaxed routing in $M_T$ requires total memory of $\Omega(\log^2{n}/\log\log{n})$ bits per vertex. Since every routing scheme in an instance $M_T$ is also a relaxed routing scheme in the same instance, our lower bound applies to routing in $M_T$.

\paragraph{Port numbering.}
In \cite{FG02}, the authors consider a family $\mathcal{T}$ of instances  where all the trees are isomorphic to $T_0$ and each instance correspond to a different port numbering of $T_0$. 
We consider a family of instances $\mathcal{T}'$ where every metric is isomorphic to $M_T$ and there is a \emph{one-to-one correspondence} between instances in $\mathcal{T}$ and those in $\mathcal{T}'$. Consider an instance $\hat{T}_0$ from $\mathcal{T}$. We proceed to explain the port numbering in the corresponding instance $\hat{M}_T$ in $\mathcal{T}'$. 
Let $u$ be an internal vertex of $T$ at height $i$. 
Define the sets of edges $E_j$
as follows:
\begin{itemize}\itemsep=0pt
\item
For $1 \le j \le t$, let $E_j$ be the set of $q_i$ edges leading to the non-dummy descendants of $u$ in the subtree of $M_T$ rooted at the $j$th child of $u$. 
\item
Partition the edges leading to dummy leaves adjacent to $u$ into groups of size $q_i$. 
Let $E_j$ be the $j$-th group for $t + 1 \le j \le d$. 
\end{itemize}

Let $p_1, \ldots, p_d$ be the port numbers of the $d$ neighbors of $u$ in $\hat{T}_0$. Note that $p_1, \ldots, p_d$ form a permutation of numbers in $\{1, \ldots, d\}$.
Define $B_k$ as the set of integers in $[(k-1)q_i + 1,  kq_i]$. 
For $1 \le j \le d$, assign to $E_j$ port numbers from $B_{p_j}$ arbitrarily. 
Assign all the other port numbers arbitrarily. 
This concludes the description of the port numbers in $\hat{M}_T$. The following observation provides the key property of the port assignments in $\hat{M}_T$.

\begin{observation}\label{obs:ports}
Let $w_i$ be a child of $u$ and let $p_i \in \{1, \ldots, d\}$ be the port number in $\hat{T}_0$ of the edge $(u,w_i)$, as seen from $u$. Every port number $p$ of $u$ in $\hat{M}_T$ leading to a subtree rooted at $w_i$ satisfies $\ceil{p/q_i} = p_i$. 
\end{observation}

\paragraph{Routing without header rewriting.} Next, we show that header rewriting cannot reduce overall memory per vertex in ancestor-descendant routing. Consider an ancestor-descendant routing scheme $\mathcal{R}'$ which routes on top of an instance $\hat{M}_T$ in $\mathcal{T}'$. 
Let $u$ be a source vertex and $v$ a destination. Initially, the header contains only $\lbl(v)$. Let $w$ be the first vertex on the routing path from $u$ to $v$. 
Since $\mathcal{R}'$ is a valid ancestor-descendant routing scheme and $w$ is an ancestor of $v$, it is possible to route from $w$ to $v$ with $w$ as a source and $v$ as a destination. In this case, the routing algorithm commences at $w$ and the header contains only $\lbl(v)$. Since the routing scheme has stretch 1, vertex $w$ will never be visited again. In other words, rewriting the header at vertex $u$ does not help in ancestor-descendant routing.

\paragraph{Reduction to routing in trees.} Let $\hat{T}_0$ be an instance in $\mathcal{T}$ and let $\hat{M}_T$ be the corresponding instance in $\mathcal{T}'$. Our goal is to define a transformation of an ancestor-descendant routing scheme $\mathcal{R}'$ for $\hat{M}_T$ into an ancestor-descendant routing scheme $\mathcal{R}$ for $\hat{T}_0$, which uses additional $O(\log{n'}) = O(\log{n})$ bits per vertex  when restricting to query pairs that exist in $\hat{T}_0$. 
Consider an internal vertex $u$ at height $i$ in $\hat{M}_T$ and its descendant (non-dummy vertex) $v$. Let $\tbl'(u)$ be the routing table of $u$ and $\lbl'(v)$ be the label of $v$ in $\mathcal{R}'$. Define $\lbl(v) \coloneqq \lbl'(v)$ and let $\tbl(u)$ be a concatenation of $\tbl'(u)$ and a binary encoding of $q_i$. The number of bits required to store $q_i$ is $O(\log n') = O(\log n)$.
Let $\mathcal{R}(\tbl(u), \lbl(v)) \coloneqq \ceil{\frac{\mathcal{R}'(\tbl'(u), \lbl'(v))}{q_i}}$. This concludes the description of $\mathcal{R}$.

We argue that $\mathcal{R}$ is a valid routing scheme for $T$. It suffices to prove that $\mathcal{R}$ outputs the correct port number. Let $p$ be the port number leading to the next vertex on the routing path from $u$ down to $v$. We want to prove that $\mathcal{R}(\tbl(u), \lbl(v)) = p$. From \Cref{obs:ports}, we know that $\ceil{\frac{\mathcal{R}'(\tbl'(u), \lbl'(v))}{q_i}} = p'$, where the $p'$ is the port number in $\hat{T}_0$ leading to the child of $u$ which is the root of the subtree where $v$ belongs. Hence, the routing algorithm in $\hat{T}_0$ proceeds at the correct child.

In conclusion, we described a way to convert a routing scheme for $\hat{M}_T$ into a routing scheme in $\hat{T}_0$ which uses $\Omega(\log{n})$ additional bits. In \cite{FG02} it is proved that $\mathcal{T}$ contains an instance in which some vertex requires $\Omega(\log^2{n}/\log\log{n})$ bits of memory. This means that there is an instance in $\mathcal{T}'$ which requires $\Omega(\log^2{n}/\log\log{n}) = \Omega(\log^2{n'}/\log\log{n'})$ bits of memory.

\paragraph{Acknowledgement.} Hung Le and Cuong Than are supported by NSF grant CCF-2517033 and NSF CAREER Award CCF-2237288. Cuong Than is also supported by a Google Ph.D. Fellowship. Shay Solomon is funded by the European Union (ERC, DynOpt, 101043159). Views and opinions expressed are however those of the author(s) only and do not necessarily reflect those of the European Union or the European Research Council. Neither the European Union nor the granting authority can be held responsible for them. Lazar Milenkovi\'{c} and Shay Solomon are funded by a grant from the United States-Israel Binational Science Foundation (BSF), Jerusalem, Israel, and the United States National Science Foundation (NSF). 
\bibliographystyle{alpha}
\bibliography{references, ramsey, RPTALGbib}

\end{document}